\RequirePackage{fix-cm}
\documentclass[a4paper,british,abstract=on,BCOR=0mm,headinclude]{scrartcl}
\usepackage{charter}
\usepackage[T1]{fontenc}
\usepackage[latin9]{inputenc}
\usepackage{babel}
\usepackage{verbatim}
\usepackage{prettyref}
\usepackage{mathtools}
\usepackage{enumitem}
\usepackage{amsmath}
\usepackage{amsthm}
\usepackage{amssymb}
\usepackage{setspace}
\usepackage[unicode=true,
 bookmarks=true,bookmarksnumbered=true,bookmarksopen=false,
 breaklinks=false,pdfborder={0 0 1},backref=false,colorlinks=false]
 {hyperref}
\hypersetup{pdftitle={A Lower Bound for Primality of Finite Languages},
 pdfauthor={Philip Sieder}}

\makeatletter

\pdfpageheight\paperheight
\pdfpagewidth\paperwidth

\numberwithin{equation}{section}
\numberwithin{figure}{section}
\theoremstyle{remark}
\newtheorem*{notation*}{\protect\notationname}
\theoremstyle{plain}
\newtheorem{thm}{\protect\theoremname}
\theoremstyle{definition}
\newtheorem{defn}[thm]{\protect\definitionname}
\theoremstyle{remark}
\newtheorem*{rem*}{\protect\remarkname}
\theoremstyle{plain}
\newtheorem{cor}[thm]{\protect\corollaryname}
\theoremstyle{plain}
\newtheorem{lem}[thm]{\protect\lemmaname}
\theoremstyle{plain}
\newtheorem{prop}[thm]{\protect\propositionname}
\newcommand\thmsname{\protect\theoremname}
\newcommand\nm@thmtype{theorem}
\theoremstyle{plain}

\newenvironment{namedthm}[1][Undefined Theorem Name]{
  \ifx{#1}{Undefined Theorem Name}\renewcommand\nm@thmtype{theorem*}
  \else\renewcommand\thmsname{#1}\renewcommand\nm@thmtype{namedtheorem}
  \fi
  \begin{\nm@thmtype}}
  {\end{\nm@thmtype}}
\theoremstyle{remark}
\newtheorem*{acknowledgement*}{\protect\acknowledgementname}

\usepackage[type1]{biolinum} 
\usepackage[type1]{libertineMono} 
\newrefformat{prop}{Proposition \ref{#1}}
\theoremstyle{definition}
  \newtheorem{problemAMS}[thm]{\protect Problem}

\usepackage{tikz}
\usetikzlibrary{automata,positioning}
\usepackage{environ}

\newlength{\outpt}   
\NewEnviron{problem}[1]{%
\begin{minipage}{\textwidth}
\begin{problemAMS}~
\begin{center}
{\centering \fbox{\parbox{.9\textwidth}{%
    {\centering\scshape #1\par}%
    \parskip=1ex
    \everypar{\hangindent=\the\outpt}%
\hspace{0cm}\rlap{Input:}\hspace{\outpt}\BODY
    
}}}\end{center}
\end{problemAMS}

\end{minipage}\medskip{}

}
\addtokomafont{disposition}{\boldmath}
\KOMAoptions{DIV=last}

\makeatother

\providecommand{\acknowledgementname}{Acknowledgement}
\providecommand{\corollaryname}{Corollary}
\providecommand{\definitionname}{Definition}
\providecommand{\lemmaname}{Lemma}
\providecommand{\notationname}{Notation}
\providecommand{\propositionname}{Proposition}
\providecommand{\remarkname}{Remark}
\providecommand{\theoremname}{Theorem}

\begin{document}
\title{A Lower Bound for Primality\\
of Finite Languages\vspace{0.8em}}
\author{Philip Sieder}

\maketitle
\global\long\def\amss{\SwapAboveDisplaySkip}%
\global\long\def\npcomp{\mathrm{NP}\text{\,-\,}complete}%
\global\long\def\np{\mathrm{NP}}%
\global\long\def\conp{\mathrm{coNP}}%
\global\long\def\nphard{\mathrm{NP}\text{\,-\,}hard}%
\global\long\def\pit{\Pi_{2}^{\mathrm{P}}}%
\global\long\def\sigt{\Sigma_{2}^{\mathrm{P}}}%
\global\long\def\pspacecomp{\mathrm{PSPACE}\text{\,-\,}complete}%
\global\long\def\conphard{\mathrm{coNP}\text{\,-\,}hard}%
\global\long\def\primfin{\mathrm{Primality_{finite}}}%
\global\long\def\primreg{\mathrm{Primality_{regular}}}%
\global\long\def\nfa{\mathrm{NFA}}%
\global\long\def\dfa{\mathrm{DFA}}%
\global\long\def\sqtil{\mathrm{SquareTiling_{rel}}}%
\global\long\def\sqtile{\mathrm{SquareTiling_{edge}}}%
\global\long\def\coneq{\mathrm{ConcatenationEquivalence_{finite}}}%
\global\long\def\lao{L_{1}}%
\global\long\def\lat{L_{2}}%
\global\long\def\lan{L}%
\global\long\def\alph{\Sigma}%
\global\long\def\til{\Theta}%
\global\long\def\nat{\mathbb{N}}%
\global\long\def\integ{\mathbb{Z}}%
\global\long\def\on#1{\mathcal{O}(#1)}%
\global\long\def\state#1{\varsigma_{#1}}%
\global\long\def\stat#1{\sigma_{#1}}%
\mathtoolsset{centercolon}
\global\long\def\emptyword{\varepsilon}%
 
\global\long\def\part{P}%
\global\long\def\parlan#1#2{L_{#1}^{#2}}%
\global\long\def\tilel#1#2#3#4{
\begin{tikzpicture}[scale=.45,baseline=-.8mm]
	\node[] (v1) at (-1,1) {};
	\node[] (v2) at (1,1) {};
	\node[] (v3) at (1,-1) {};
	\node[] (v4) at (-1,-1) {};
	\draw (v1.center) --node[below,inner sep=2] {$\scriptstyle #1$} (v2.center) -- node[left,inner sep=2](b1) {$\scriptstyle  #2$} (v3.center) -- node[above,inner sep=2] {$\scriptstyle  #3$}  (v4.center) -- node[right,inner sep=2] {$\scriptstyle  #4$} (v1.center) --(v2.center);
	\draw (v1.center) -- (v3.center)  (v2.center) -- (v4.center);
	\end{tikzpicture}
}

\global\long\def\tile#1#2#3#4{\tilel{#1}{#2}{#3}{#4}}%
\global\long\def\disun{\dot{{}\cup{}}}%
\global\long\def\sep{\mid}%
\global\long\def\al{A_{l}}%
\global\long\def\ar{A_{r}}%

\begin{abstract}
\noindent A regular language $\lan$ is said to be prime, if it is
not the product of two non-trivial languages. Martens et al.\ settled
the exact complexity of deciding primality for deterministic finite
automata in 2010. For finite languages, Mateescu et al.\ and Wieczorek
suspect the $\npcomp ness$ of primality, but no actual bounds are
given. Using the techniques of Martens et al., we prove the $\np$
lower bound and give a $\pit$ upper bound for deciding primality
of finite languages given as deterministic finite automata.
\end{abstract}
\tableofcontents{}

\newpage{}

\section{Introduction}

Coming from number theory, the primality of regular languages is a
quite natural problem. As integers have a unique prime factorisation,
one could hope to decompose languages into indecomposable (and therefore
possibly simpler) languages. Unfortunately the decompositions of languages
do not behave as nicely as those of numbers. A language, if decomposable,
can have different decompositions. Neither the number of prime factors
is unique nor do different decompositions need to have common prime
factors \cite[Section 4]{MSY98}. Therefore the most interesting question
is, whether a language can be decomposed at all, or in other words
whether a language is prime. As in number theory, the complexity of
a primality test (for regular languages) was pinpointed relatively
recently. Martens et al.\ \cite{MNS10} showed that the problem is
$\pspacecomp$. For finite languages in particular, there are pursuits
by Mateescu et al.\ \cite{MSY98} and Wieczorek \cite{Wiec10}, but,
besides an $\npcomp ness$ conjecture, no actual bounds have been
given. Using the ideas of Martens et al., we prove an $\np$ lower
bound and a $\pit$ upper bound for the problem. So again languages
behave way worse then numbers, where primality can be tested in polynomial
time.

In \prettyref{sec:Preliminaries} we establish the notation and give
definitions for the general language theoretical facts we need. In
\prettyref{sec:intoPrim} we give some insight on the necessary properties
for studying primality of regular languages. Those enable the proof
of the $\pit$ upper bound at the end of the section. \prettyref{sec:NPhard}
provides the $\nphard ness$ by establishing a chain of polynomial
time reductions, similar to the one in the proof of Martens, Niewerth
and Schwentick. In the final \prettyref{sec:Final-remarks}, we give
a brief compilation of what is yet to be determined.

\section{Preliminaries\label{sec:Preliminaries}}

In this section we will introduce the basic concepts and notations.
We omit the facts about complexity classes and polynomial time reduction.
For those concepts and definitions we refer to Papadimitriou's book
\cite{Papa94}. First let us fix some general symbols:
\begin{notation*}
$[a,b]:=\{m\in\integ\sep a\leq m\leq b\}$ with $a,b\in\integ$ integers.
$n\integ:=\{n\cdot m\sep m\in\integ\}$ with $n\in\integ$ an integer.\\
For a computational decision problem $\mathrm{PROBLEM}$, $\neg\mathrm{PROBLEM}$
describes the same problem with negated answer.

Now we will introduce the most important concepts about regular languages
and finite automata we use. Since this part is mostly to fix the notation,
we do not give much explanation or motivation and the definitions
might have minor inaccuracies. For a more thorough understanding of
those conceptions we refer to the book of Hopcroft et al.\ \cite{Hopcroft}.
\end{notation*}
\begin{defn}
A (finite) \emph{alphabet} is a finite set $\alph$ of letters. A
\emph{word} $w=a_{1}\ldots a_{n}$ is a finite sequence of letters
$a_{i}\in\alph$ and $|w|=|a_{1}\ldots a_{n}|:=n$ is the \emph{length
of the word}. The \emph{empty word} (of length zero) is written as
$\emptyword$. For two words $v=a_{1}\ldots a_{m}$ and $w=b_{1}\ldots b_{n}$,
$v\circ w:=vw:=a_{1}\ldots a_{m}b_{1}\ldots b_{n}$ describes the
\emph{concatenation} of the two words $v$ and $w$. The \emph{Kleene
closure} of $\alph$ is $\alph^{*}:=\bigcup_{n\geq0}\alph^{n}$ where
$\alph^{n}$ denotes the set of all words over the alphabet $\alph$
with length $n$. Additionally $\alph^{+}:=\bigcup_{n\geq1}\alph^{n}$
is the set of all words with positive length. A \emph{language} $\lan\subseteq\alph^{*}$
is a set of words. A \emph{finite language} is a language containing
only finitely many words. For two languages $\lao$ and $\lat$ over
an alphabet $\alph$, the term $\lao\circ\lat:=\lao\lat:=\{vw\in\alph^{*}\sep v\in\lao\text{ and }w\in\lat\}$
describes the \emph{product} (or concatenation) of the two languages. 
\end{defn}

\begin{defn}[finite automaton]
 A \emph{nondeterministic finite automaton} ($\nfa$) $M$ is a tuple
$(Q,\alph,\delta,I,F)$ where $Q$ is a finite set of states, $\alph$
is a finite alphabet, $\delta\colon Q\times\alph\rightarrow2^{Q}$
is the transition function, $I\subseteq Q$ is the set of initial
states and $F\subseteq Q$ is the set of accepting states. The automaton
is called a \emph{deterministic finite automaton} ($\dfa$) if $|I|=1$
and for all $q\in Q$ and all $a\in\alph$ the inequation $|\delta(q,a)|\leq1$
holds.
\end{defn}

\begin{rem*}
In this thesis, if not explicitly mentioned otherwise, an ``automaton''
is a $\dfa$. \\
We allow $\delta(q,a)=\varnothing$ for $\dfa$s to simplify their
specification. To get a model where $\delta$ is a total function
one only has to add a sink state $g$ such that $\delta(q,a)$ equals
$\{g\}$ instead of $\varnothing$ and $\delta(g,a)=\{g\}$ for all
$a\in\alph$. When a transition function is defined in this paper,
a not considered pair $(q,a)\in Q\times\alph$ means $\delta(q,a)=\varnothing$.
Furthermore, if $\delta(q,a)=\{q'\}$ is a singleton, we write $\delta(q,a)=q'$. 
\end{rem*}
\begin{notation*}
Let $(Q,\alph,\delta,I,F)$ be an $\nfa$, $S\subseteq Q$, $w\in\alph^{*}$
and $a\in\alph$. Then we define
\end{notation*}
\begin{itemize}
\item $\delta(S,a):=\bigcup_{q\in S}\delta(q,a)$ 
\item $\delta(S,w)$ inductively as $\delta(S,aw):=\delta(\delta(S,a),w)$
\\
\phantom{a}\hfill (the states reached from $S$ after reading $w$)
\item $\delta^{*}(S,w)$ inductively as $\delta^{*}(S,aw):=\delta(S,a)\cup\delta(\delta(S,a),w)$
\\
~\phantom{a}\hfill (all states visited from $S$ by reading $w$)
\end{itemize}
If $S=\{q\}$ is a singleton, we write $\delta(q,w)$ and $\delta^{*}(q,w)$.
\begin{defn}
Let $M=(Q,\alph,\delta,I,F)$ be an $\nfa$. \\
The language $L(M):=\{w\in\alph^{*}\sep\delta(I,w)\cap F\neq\varnothing\}$
is the language defined by $M$.\\
A language $\lan\subseteq\alph^{*}$ is called \emph{regular}, if
there is an $\nfa$ $M$ such that $\lan=L(M)$.
\end{defn}

\begin{rem*}
Every regular language $\lan$ has a $\dfa$ $M$ such that $\lan=L(M)$.
\end{rem*}
\begin{cor}
Every finite language is regular.
\end{cor}

\section{An introduction to primality of regular languages\label{sec:intoPrim}}

In this section we give the definitions, important properties and
known results about the primality of regular and finite languages.
First of we start with a definition of primality.
\begin{defn}[Primality]
A regular language $\lan\subseteq\alph^{*}$ is called \emph{decomposable},
if there are languages $\lao,\lat\subseteq\alph^{*}$, $\lao\neq\{\emptyword\}\neq\lat$
such that $\lan=\lao\circ\lat$. If $\lan$ is not decomposable it
is called \emph{prime}. 
\end{defn}

\begin{rem*}
As we see in \prettyref{thm:partitionSets}, it makes no difference
whether we require $\lao$ and $\lat$ to be regular languages.
\end{rem*}
The definition adverts the following decision problem:

\begin{problem}{$\primreg$}
A regular language $\lan$ over a finite alphabet $\alph$ given as a $\dfa$

Question: Is $\lan$ prime
\end{problem}The exact complexity of this problem was determined relatively recently:
\begin{thm}[{\cite[Corollary 6.10]{MNS10}}]
$\primreg$ is $\pspacecomp$.
\end{thm}

For finite languages the exact complexity of the problem is not yet
known. To the best of our knowledge, the $\nphard ness$, which we
prove in \prettyref{thm:PrimNPhard}, was not known before. Let us
start with a definition of the problem.

\begin{problem}{$\primfin$}
A finite language $\lan$ over a finite alphabet $\alph$ given as a $\dfa$

Question: Is $\lan$ prime
\end{problem}

The problem was examined before: The paper of Mateescu et al.\ \cite{MSY98}
establishes some notions, gives general results and treats examples.
They suspect $\npcomp ness$ for $\primfin$, but only give a double
exponential algorithm \cite[Theorem 3.1 and below]{MSY98}. A less
theoretical approach takes Wieczorek \cite{Wiec10}, as he offers
an optimised deterministic algorithm for a finite language given as
a list. If the finite language is given as a list of words, the primality
problem is obviously in $\conp$: One guesses a partition in two parts
for every word and checks whether all combinations of a first part
of one and a second part of another word are again in the given language.
As the description of a finite language as a list can be exponentially
larger than the corresponding $\dfa$ (for instance the language of
all words of a specific length), the algorithm is not useful for our
problem.

To check for primality of a language $\lan$, one has to consider
if there are languages $\lao$ and $\lat$ that decompose $\lan=\lao\lat$.
Because we have to work with the $\dfa$ of $\lan$, we should examine
the states in which the words get actually split. That leads to the
following definition and results:
\begin{defn}
Let $\lan$ be a regular language, given as a $\dfa$ $M=(Q,\alph,\delta,\{s\},F)$
and $\part\subseteq Q$ a set of states. We call $\part$ a \emph{partition
set} and define the regular languages
\[
\parlan 1{\part}:=\{w\in\alph^{*}\sep\delta(s,w)\in\part\}
\]
 and 
\[
\parlan 2{\part}:=\bigcap_{p\in P}\{w\in\alph^{*}\sep\delta(p,w)\in F\}.
\]
\end{defn}

\begin{rem*}
The languages $\parlan 1{\part}$ and $\parlan 2{\part}$ are regular
because $(Q,\alph,\delta,\{s\},\part)$ is an automaton for $\parlan 1{\part}$
and $(Q,\alph,\delta,\{p\},F)$ is an automaton for $\{w\in\alph^{*}\sep\delta(p,w)\in F\}$
and an intersection of regular languages is regular again \cite[Section 4.2]{Hopcroft}. 
\end{rem*}
\begin{lem}
\label{lem:checkLinL1L2}Let $\lan$ be a regular language given as
a $\dfa$ $M=(Q,\alph,\delta,\{s\},F)$ and let $\part\subseteq Q$
be any subset, then $\parlan 1{\part}\parlan 2{\part}\subseteq\lan$.
\end{lem}

\begin{proof}
Let $w_{1}w_{2}\in\parlan 1{\part}\parlan 2{\part}$ with $w_{i}\in\parlan i{\part}$,
then $\delta(s,w_{1})\in\part$ by the definition of $\parlan 1{\part}$
and therefore $\delta(s,w_{1}w_{2})=\delta(\delta(s,w_{1}),w_{2})\in F$
by the definition of $\parlan 2{\part}$.
\end{proof}
\begin{thm}[{\cite[Lemma 3.1]{MSY98}}]
\label{thm:partitionSets}Let $\lan$ be a regular language, given
as a $\dfa$ $M=(Q,\alph,\delta,\{s\},F)$, let $\lan=\lao\lat$ be
a decomposition of $\lan$ and let 
\[
\part:=\{q\in Q\sep q=\delta(s,w)\text{ for some }w\in\lao\}
\]
 be the set of ``border''-states. Then $\lao\subseteq\parlan 1{\part}$,
$\lat\subseteq\parlan 2{\part}$ and 
\[
L=\parlan 1{\part}\parlan 2{\part}
\]
 is the decomposition of $\lan$ into two regular languages.
\end{thm}

\begin{proof}
$\lao\subseteq\parlan 1{\part}$: Let $w\in\lao$, then $\delta(s,w)\in\part$
and therefore $w\in\parlan 1{\part}$.\\
$\lat\subseteq\parlan 2{\part}$: Suppose $w\in\lat\setminus\parlan 2{\part}$,
that means $w\in\lat$ and there is a $p\in\part$ such that $\delta(p,w)\notin F$.
Let $v\in\lao$ such that $\delta(s,v)=p$. Then $vw\notin\lan$,
because $\delta(s,vw)=\delta(\delta(s,v),w)=\delta(p,w)\notin F$,
but at the same time $vw\in\lao\lat=L$. That contradicts the existence
of $w\in\lat\setminus\parlan 2{\part}$.\\
$L=\parlan 1{\part}\parlan 2{\part}$: The inclusion $\lan\subseteq\parlan 1{\part}\parlan 2{\part}$
follows directly from $L=\lao\lat$ and $L_{i}\subseteq\parlan i{\part}$
for $i\in\{1,2\}$.%
{} The other inclusion was given in \prettyref{lem:checkLinL1L2}.
\end{proof}
The theorem enables us to limit our search for decompositions to the
ones that arise from this construction. The problem is, after guessing
a partition set $\part$, to actually check whether $\lan\subseteq\parlan 1{\part}\parlan 2{\part}$.
Unfortunately the intersection of $\on n$ sets and the concatenation
of two languages is not efficient, as both can lead to an exponential
blow-up of the number of states.

We do not use the following theorem from Wieczorek \cite{Wiec10},
which is included for readers interested in further research. It allows
to reduce the states that have to be considered for $\part$, but
a reduction beyond $\on n$ is neither obvious nor likely.
\begin{thm}[{\cite[Theorem 3]{Wiec10}}]
 Let $\lan$ be a decomposable finite language with a minimal $\dfa$
$M=(Q,\alph,\delta,\{s\},F)$. Then there is a partition set $\part$
with $\lan=\parlan 1{\part}\parlan 2{\part}$ such that for all $p\in\part$
either $|\{a\in\alph\sep\delta(p,a)\}|>1$ or $(p\in F)\wedge(\exists w\in\alph^{*}\colon\delta(p,w)\in F)$
holds.
\end{thm}

Unfortunately we did not close the gap between the $\np$ lower and
the $\pit$ upper bound. But let us at least provide a proof for the
$\pit$ upper bound:
\begin{prop}
$\primfin$ is in $\pit$.
\end{prop}

\begin{proof}
The definitions for the polynomial hierarchy can be found in Papadimitriou's
book \cite[Section 17.2]{Papa94}. We will argue that $\neg\primfin$
is in $\Sigma_{2}^{\mathrm{P}}$ by the characterisation of \cite[Chapter 17, Corollary 2]{Papa94}:
\begin{multline*}
\neg\primfin=\\
\{\lan=L(Q,\alph,\delta,I,F)\mid\exists\part\subseteq Q\ \forall w\in L\colon(L,P,w)\in R\,:\Longleftrightarrow\,w\in\parlan 1{\part}\parlan 2{\part}\}
\end{multline*}
Using \prettyref{thm:partitionSets} and \prettyref{lem:checkLinL1L2},
the right side is a characterisation of $\neg\primfin$. We have to
check that the relation $R$ is polynomial-time decidable and is polynomially
balanced. For a finite language $\lan$ let $M=(Q,\alph,\delta,\{s\},F)$
be the $\dfa$ of $\lan$ and $n$ its size. The relation is polynomial-time
decidable: One simulates $M$ on the input $w$ and stores the set
$\part_{w}:=\delta^{*}(s,w)\cap\part$ and the remaining characters
of $w$ (when reaching $p\in\part_{w}$) in $W\subseteq\alph^{*}$.
If $\part_{w}=\varnothing$, we reject. Otherwise we simulate for
all $v\in W$ and all $p\in\part_{w}$ the automaton $M_{p}:=(Q,\alph,\delta,\{p\},F)$
on $v$. If there is at least one $v$ such that $v\in L(M_{p})$
for all $p\in\part_{w}$, we accept or else we reject. So the test
takes at most time $\on{n+n\cdot n}$. \\
The relation is polynomially balanced as well since the partition
set has at most $n$ elements and $w$ has length at most $n-1$ ($M$
is acyclic since the language is finite).
\end{proof}

\section{$\protect\nphard ness$ of $\protect\primfin$\label{sec:NPhard}}

In this chapter we proof the following main theorem of the paper:
\begin{thm}
\label{thm:PrimNPhard}$\primfin$ is $\nphard$ (for languages given
as $\dfa$s).
\end{thm}

We will start with the $\npcomp$ problem $\sqtile$ and build the
following chain of polynomial reductions:
\begin{multline*}
\mathrm{\np}\leq\sqtile\leq\sqtil\leq\\
\neg\coneq\leq\primfin
\end{multline*}

The chain is actually quite similar to the one in the work of Martens
et al.\ \cite[Sections 5.2 and 6.2]{MNS10}. They reference a different
form of tiling and use a special case of concatenation equivalence.

\subsection{From $\protect\sqtile$ to $\protect\sqtil$}

We start with a tiling problem whose complexity is stated in the book
of Garey and Johnson \cite{GJ79}. Then we will adapt the problem
to a better fitting variant.

\begin{problem}{$\sqtile$}
A set of colours $C$, a set of tiles $\mathcal{T} \subseteq C^4$ and a natural number ${n \leq |C|}$; \\
A tile $\tile{a}{b}{c}{d}:=(a,b,c,d)\in \mathcal{T}$ has four edges with corresponding colours

Question: Is there a tiling, i.e.\ an $n\times n$ square $A\in \mathcal{T}^{n\times n}$ of tiles, such that all adjacent tiles $A(i,j)=\tile{a}{\boldsymbol{b}}{\boldsymbol{c}}{d}$ and $A(i,j+1)=\tile{\alpha}{\beta}{\gamma}{\boldsymbol{\delta}}$ resp.\ $A(i+1,j)=\smash[t]{\tile{\boldsymbol{\tilde a}}{\tilde b}{\tilde c}{\tilde d}}$ fullfill $b=\delta$ resp.\ $c=\tilde a$
\end{problem}
\begin{prop}[{\cite[GP13]{GJ79}\footnote{The source only mentions that the directed-hamilton-path problem is
reduced to $\sqtile$. To give the interested reader a basis for the
proof: $n:=\#\text{vertices}$, the series of vertices in the hamilton
path is written on the diagonal of the $n\times n$-square and the
corresponding edges are on the diagonals above and below the main
diagonal.}}]
$\sqtile$ is $\npcomp$.
\end{prop}

\begin{problem}{$\sqtil$}
A set of tiles $\til$, relations $V, H\subseteq\til\times\til$ and a natural number $n\in\nat$

Question: Is there a tiling, i.e.\ an $n\times n$ square $T\in \til^{n\times n}$, such that adjacent tiles are in the horizontal relation $H$ resp.\ the vertical relation $V$:
\begin{align*}
\forall i \ \forall j<n\colon & (T(i,j),T(i,j+1))\in H\\ 
\forall i<n \ \forall j\colon & (T(i,j),T(i+1,j))\in V 
\end{align*}
\end{problem}

\begin{rem*}
Alternatively we write $T(i\cdot n+j):=T(i,j)$ and get a list where
$(T(m),T(m+1))\in H$ for $1\leq m<n^{2}\wedge m\notin n\integ$ and
$(T(m),T(m+n))\in V$ for $1\leq m\leq n^{2}-n$ has to be fulfilled. 
\end{rem*}
\begin{prop}
\label{prop:sqTilRelNP}$\sqtil$ is $\nphard$.
\end{prop}

\begin{proof}
Given an input $C$, $\mathcal{T}$ and $n$ for $\sqtile$. Let 
\begin{align*}
H:= & \{(\tile a{\boldsymbol{b}}cd,\tile{\alpha}{\beta}{\gamma}{\boldsymbol{\delta}})\in\mathcal{T}\times\mathcal{T}\sep b=\delta\}\text{,}\\
V:= & \{(\tile ab{\boldsymbol{c}}d,\tile{\boldsymbol{\alpha}}{\beta}{\gamma}{\delta})\in\mathcal{T}\times\mathcal{T}\sep c=\alpha\}
\end{align*}
and $\til:=T$. Then there is a tiling $T$ for $\sqtil(\til,H,V,n)$
if and only if there is one for $\sqtile(C,\mathcal{T},n)$. The construction
of $\til$, $H$ and $V$ works obviously in polynomial time.
\end{proof}
\begin{rem*}
One can translate $\sqtil$ to $\sqtile$ as well, as outlined in
a paper of van Emde Boas \cite[p.~7]{Van97}.
\end{rem*}

\subsection{From $\protect\sqtil$ to $\protect\coneq$ }

This is the most interesting reduction in the chain. Here a truly
original idea, not present in the proof of Martens et al.\ \cite{MNS10},
is necessary.

\begin{problem}{$\coneq$}
Finite languages $\lan$, $\lao$ and $\lat$ over a finite alphabet $\alph$ given as $\dfa$s

Question: Does $\lan = \lao \lat$ hold
\end{problem}

Now we will reduce $\sqtil$ to $\neg\coneq$. The most interesting
point, compared to the regular language case, is that we work over
the alphabet $\til\times[1,n^{2}]$ instead of just $\til$. This
allows us, for a word in $\lao\lat$, to detect the point where we
jump from $\lao$ to $\lat$.
\begin{prop}
\label{prop:ConEqNP}$\coneq$ is $\conp\text{\,-\,}complete$.
\end{prop}

\begin{proof}
It is obviously in $\conp$, since a word in $\lan\setminus\lao\lat$
or in $\lao\lat\setminus\lan$ is a witness for $\lan\neq\lao\lat$
and the longest words to consider have length $\on n$, as the languages
are finite. So we get to the $\conphard ness$. Suppose we can solve
$\coneq$. Let $n$, $\til$ and $V,H\subseteq\til\times\til$ be
an input for $\sqtil$. We define 
\begin{align*}
\lao:={} & \vphantom{\bigcup_{n^{2}\integ}}\{(t_{1},1)(t_{2},2)\ldots(t_{m},m)\in(\til\times[1,n^{2}])^{*}\sep m\leq n^{2}-2\}\text{,}\\
\lat:={} & \enskip\bigcup_{\mathclap{1\leq m\leq n^{2},m\notin n\integ}}\enskip\{(t_{m},m)(t_{m+1},m+1)\ldots(t_{n^{2}},n^{2})\in(\til\times[1,n^{2}])^{*}\sep(t_{m},t_{m+1})\notin H\}\,\cup\\
 & \enskip\bigcup_{\mathclap{1\leq m\leq n^{2}-n}}\enskip\{(t_{m},m)(t_{m+1},m+1)\ldots(t_{n^{2}},n^{2})\in(\til\times[1,n^{2}])^{*}\sep(t_{m},t_{m+n})\notin V\}
\end{align*}
and
\[
\lan:=\lao\lat\cup\{(t_{1},1)(t_{2},2)\ldots(t_{n^{2}},n^{2})\in(\til\times[1,n^{2}])^{n^{2}}\}.
\]
The size of the $\dfa$s of the defined languages is polynomial in
the size of the input and can be constructed in polynomial time as
shown below.

\begin{samepage}
The automaton for $L_{1}$ is pretty simple and has $n^{2}-1$ states:
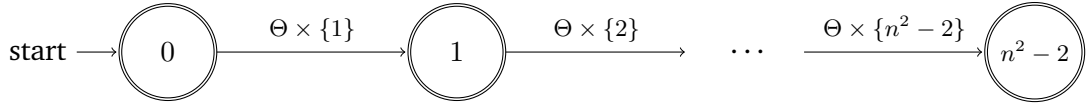
\begin{center}
\tikzset{every state/.style={minimum size=1.25cm}}\begin{tikzpicture}[shorten >=1pt,node distance=3.8,on grid,auto]
\node[state,initial,accepting] (q_0) {$0$};
\node[state,accepting] (q_1) [right=of q_0] {$1$};
\node[] (q_2) [right=of q_1] {~~~~\ldots~~~~};
\node[state,accepting] (q_n2) [right=of q_2] {\footnotesize{$n^2-2$}};
\path[->]
(q_0) edge node {\footnotesize{$\til \times \{1\}$}} (q_1)
(q_1) edge node {\footnotesize{$\til \times \{2\}$}} (q_2)
(q_2) edge node {\footnotesize{$\til \times \{n^2-2\}$}} (q_n2);
\end{tikzpicture}
\end{center}
\end{samepage}The automaton for $\lat$ is more complicated and depends on $V$
and $H$, but it is polynomial in size. We will give two automata
$M_{V}$ and $M_{H}$ with polynomial sizes such that 
\[
L(M_{H})=\enskip\bigcup_{\mathclap{1\leq m\leq n^{2},m\notin n\integ}}\enskip\{(t_{m},m)(t_{m+1},m+1)\ldots(t_{n^{2}},n^{2})\in(\til\times[1,n^{2}])^{*}\sep(t_{m},t_{m+1})\notin H\}
\]
 and 
\[
L(M_{V})=\enskip\bigcup_{\mathclap{1\leq m\leq n^{2}-n}}\enskip\{(t_{m},m)(t_{m+1},m+1)\ldots(t_{n^{2}},n^{2})\in(\til\times[1,n^{2}])^{*}\sep(t_{m},t_{m+n})\notin V\}.
\]
Obviously $\lat=L(M_{H}\cup M_{V})$ and the union automaton still
has polynomial size \cite[Proof of Theorem 2.1]{Yu1997}. 

The automaton $M_{H}$ is constructed as follows: The set of states
is
\[
Q_{H}:=\{s_{H}\}\disun\{\state{t,m}\sep t\in\til,m\in[1,n^{2}]\setminus n\integ\}\disun\{\varsigma_{i}\sep i\in[2,n^{2}]\}.
\]
The automaton has to check for a word $(t_{1},m_{1})(t_{2},m_{2})\ldots(t_{k},m_{k})$
whether $(t_{1},t_{2})\notin H$ and whether $m_{i+1}=m_{i}+1$ for
all $i$. So after reading the first letter the state has to store
$t_{1}$ and every state has to store the most recent $m_{i}$. Therefore
after the first character we go to the corresponding state $\state{t_{1},m_{1}}$.
If the next character fulfils both $(t_{1},t_{2})\notin H$ and $m_{2}=m_{1}+1$,
we only have to check $m_{i+1}=m_{i}+1$. Hence we only store the
most recent $m_{i}$, by going to the state $\state{m_{i}}$. Once
we get to $\state{n^{2}}$ we accept. If otherwise there was any mistake
we stop the run at that point. 

Here a formal definition of the transition function $\delta_{H}$:
\begin{alignat*}{2}
\delta{}_{H}(s_{H},(t,m)):={} & \state{t,m} &  & \;\text{ for }1\leq m<n^{2}\text{ and }m\notin n\integ\\
\delta_{H}(\state{t,m},(t',m+1)):={} & \state{m+1} &  & \text{\; for }1\leq m<n^{2}\text{ and }(t,t')\notin H\\
\delta_{H}(\state m,(t,m+1)):={} & \state{m+1} &  & \;\text{ for }1<m<n^{2}
\end{alignat*}
The automaton is then defined as $M_{H}:=(Q_{H},\til\times[1,n^{2}],\delta_{H},\{s_{H}\},\{\state{n^{2}}\})$
and has $1+|\til|\cdot(n^{2}-n)+n^{2}-1$ states.

The automaton $M_{V}$ is quite similar. The only difference is, that
we have to check for a word $(t_{1},m_{1})(t_{2},m_{2})\ldots(t_{k},m_{k})$,
whether $(t_{1},t_{1+n})\notin V$. Therefore we need the additional
states $\stat{t,m,o}$, where $o$ stores how many characters away
from $(t_{1},m_{1})$ we already are. Hence we get the following set
of states
\[
Q_{V}:=\{s_{V}\}\disun\{\stat{t,m,o}\sep t\in\til,m\in[1,n^{2}-n],o\in[0,n-1]\}\disun\{\stat i\sep i\in[n+1,n^{2}]\}\text{,}
\]
 the transition function
\begin{alignat*}{2}
\delta{}_{V}(s_{H},(t,m)):={} & \stat{t,m,0} &  & \;\text{ for }1\leq m\leq n^{2}-n\\
\delta_{V}(\stat{t,m,i},(t',m+i+1)):={} & \stat{t,m,i+1} &  & \;\text{ for }0\leq i<n-1\\
\delta_{V}(\stat{t,m,n-1},(t',m+n)):={} & \stat{m+n} &  & \;\text{ for }1\leq m\leq n^{2}-n\text{ and }(t,t')\notin V\\
\delta_{V}(\stat m,(t,m+1)):={} & \stat{m+1} &  & \;\text{ for }1\leq m<n^{2}
\end{alignat*}
and finally the $\dfa$ is given as $M_{V}:=(Q_{V},\til\times[1,n^{2}],\delta_{V},\{s_{V}\},\{\stat{n^{2}}\})$
with $1+|\til|\cdot(n^{2}-n)\cdot n+n^{2}-n$ states. 

So at last we have to show that $L=\lao\lat\cup\{(t_{1},1)(t_{2},2)\ldots(t_{n^{2}},n^{2})\in(\til\times[1,n^{2}])^{n^{2}}\}$
has a polynomial-sized automaton. Using this description, we might
get an exponential blow-up from the concatenation, but (with the help
from the $[1,n^{2}]$ part of the alphabet) the language can be characterised
a bit differently. It basically contains all properly numbered tilings
and additionally those with one jump and a forbidden tiling (with
a fault, either vertically or horizontally, directly after the jump).
An automaton for this can be constructed using a $\dfa$ $M_{2}=(Q_{2},\til\times[1,n^{2}],\delta_{2},\{s_{2}\},F_{2})$
that accepts $L_{2}$.

As set of states we use $Q:=Q_{2}\setminus\{s_{2}\}\disun[0,n^{2}]$
and the transition function is as follows
\begin{alignat*}{1}
\delta(q,(t,m)):={} & \begin{cases}
\begin{cases}
q+1 & m=q+1,0\leq q<n^{2}\\
\delta_{2}(s_{2},(t,m)) & m\neq q+1
\end{cases} & q\in[0,n^{2}]\\
\delta_{2}(q,(t,m)) & q\in Q_{2}\setminus\{s_{2}\}
\end{cases}
\end{alignat*}
The idea is to check for legal numbering with the states $[0,n^{2}]$.
If there is a leap in the numbering, we jump into the automaton for
$\lat$. So the automaton is given by $M:=(Q,\til\times[1,n^{2}],\delta,\{0\},\{n^{2}\}\cup F_{2})$.
Obviously $L(M)=\lan$ holds and $M$ has polynomial size.

Thus $\dfa$s for $\lao$, $\lat$ and $\lan$ are constructed in
polynomial time and have polynomial size in the size of the tiling
problem. Combining this with the following \prettyref{lem:concatSqrTilconstructedLanguages}
yields a reduction from $\sqtil$ to $\neg\coneq$. That $\sqtil$
is $\nphard$ (\prettyref{prop:sqTilRelNP}) completes the proof.
\end{proof}
\begin{lem}
\label{lem:concatSqrTilconstructedLanguages}Let $n$, $\til$, $V,H\subseteq\til\times\til$
be an input for $\sqtil$ and $\lao$, $\lat$ and $\lan$ constructed
as above. Then $\lan=\lao\lat$ if and only if there is no legal tiling.
\end{lem}

\begin{proof}
A word in $\{(t_{1},1)(t_{2},2)\ldots(t_{n^{2}},n^{2})\in(\til\times[1,n^{2}])^{n^{2}}\}$
can be interpreted as a tiling, where $T(j)=t_{j}$. Every word $w_{1}w_{2}\in\{(t_{1},1)(t_{2},2)\ldots(t_{n^{2}},n^{2})\in(\til\times[1,n^{2}])^{n^{2}}\}\cap\lao\lat$
with $w_{i}\in\lan_{i}$ represents a tiling that violates the given
relations: Let $w_{2}=(t_{m},m)(t_{m+1},m+1)\ldots(t_{n^{2}},n^{2})\in\lat$,
then either $(t_{m},t_{m+1})\notin H$ or $(t_{m},t_{m+n})\notin V$
which contradicts a legal tiling.

Let $\lan=\lao\lat$, then $\{(t_{1},1)(t_{2},2)\ldots(t_{n^{2}},n^{2})\in(\til\times[1,n^{2}])^{n^{2}}\}\subset\lao\lat$,
so every possible tiling violates the relations and therefore there
is no legal tiling. On the other hand, if there is no legal tiling,
then every possible tiling violates a relation. Hence $\{(t_{1},1)(t_{2},2)\ldots(t_{n^{2}},n^{2})\in(\til\times[1,n^{2}])^{n^{2}}\}\subset\lao\lat$
which yields $\lan=\lao\lat$.
\end{proof}

\subsection{From $\protect\coneq$ to $\protect\primfin$}
\begin{namedthm}[\prettyref{thm:PrimNPhard}]
$\primfin$ is $\nphard$.
\end{namedthm}
The following proof is similar to \cite[Proof of Theorem 6.4]{MNS10}.
The difference is, since they treat (non-finite) regular languages,
that they reduce the problem $\lao\lat\overset{?}{=}\alph^{*}$ (so
for them $\lan=\alph^{*}$).
\begin{proof}[Proof of \prettyref{thm:PrimNPhard}]
Let $\lao$, $\lat$ and $\lan$ be finite languages over the alphabet
$\alph$ given as $\dfa$s. We want to construct a language $A$,
such that $A$ is decomposable if and only if $\lan=\lao\lat$, which
reduces $\coneq$ to $\neg\primfin$ and proves the theorem by the
$\conphard ness$ of $\coneq$ (\prettyref{prop:ConEqNP}).

If $L=\varnothing$, $L=\{\emptyword\}$, $\lao=\varnothing$ or $\lat=\varnothing$,
then it is easy to check whether $\lan=\lao\lat$. So we can assume
$\varnothing\neq L\neq\{\emptyword\}$ and $\lao\neq\varnothing\neq\lat$.

Let $\alph':=\{a'\sep a\in\alph\}$ be a disjoint copy of the alphabet
and let $\$\notin\alph\disun\alph'$ be an additional letter. $\lao'$
and $\lat'$ are the respective languages over $\alph'$.

Now we define the language 
\[
A:=\lan\cup\lao\$\lat'\cup\lao'\$\lat\cup\lao'\$\$\lat'.
\]
The language's $\dfa$ is obviously constructable in polynomial time.
\begin{lem}
\label{lem:Adecomposition}The language $A$ is either prime or its
only non-trivial decomposition is $A_{1}\circ A_{2}$ with $A_{1}:=\lao\cup\lao'\$$
and $A_{2}:=\lat\cup\$\lat'$.
\end{lem}

\begin{rem*}
The proof of Martens et al.\ \cite[Claim 6.5]{MNS10} in the paper's
appendix works nearly word for word. It is rather technical and adds
no real value. For the sake of completeness we provide one regardless. 
\end{rem*}
\begin{proof}
Suppose $A=\al\ar$ is a non-trivial decomposition. We first show
that $A_{l}\subseteq\alph^{*}\cup\alph'^{*}\$$ and symmetrically
$A_{r}\subseteq\alph^{*}\cup\$\alph'^{*}$.

Suppose $\al$ contains a word $w_{l}$ with two $\$$-letters in
it or where a symbol from $\alph$ precedes a $\$$-sign. \\
In both cases, for $w_{l}w_{r}$ to be in $A$, $w_{r}$ has to be
in $\alph'^{*}$. Thus $\ar\subseteq\alph'^{*}$ and, since the decomposition
is non-trivial, $\ar\supsetneq\{\emptyword\}$. The language $\lan\subseteq\alph^{*}$
contains at least one word $v$ of length $\geq1$ (see premises).
So we can concatenate $v\in\lan\subseteq\al$ (because $\ar\subseteq\alph'^{*}$
and $L\subseteq A$) with a word $\emptyword\neq w\in\ar$ and should
get $vw\in\al\ar=A$. That is a contradiction, as $A$ does not contain
a word in $\alph^{+}\alph'^{+}$. 

The language $\ar$ includes at least one word $w$ containing a $\$$,
because $A\supseteq\lao'\$\$\lat'\neq\varnothing$ (we assume $\lao\neq\varnothing\neq\lat$)
and any word in $\al$ contains at most one $\$$. If any word $v\in\al$
includes a $\$$ not as its last character, we get a paradox because
$vw\in\al\ar=A$ incorporates two $\$$-signs that are not next to
each other. That cannot happen for a word in $A$.

Now we know, every word in $\al$ contains at most one $\$$-sign
and if it contains one, the $\$$ is the last sign and is not preceded
by a letter in $\alph$. Hence $\al\subseteq\alph^{*}\cup\alph'^{*}\$$.
Symmetrically (one can look at the reversed languages) $\ar\subseteq\alph^{*}\cup\$\alph'^{*}$.

The intersection $A\cap\alph^{*}\$\alph'^{*}=\lao\$\lat'$ together
with the structures of $\al$ and $\ar$ yield $\al\cap\alph^{*}=\lao$
and symmetrically $\ar\cap\alph{}^{*}=\lat$. 

Similarly $A\cap\alph'^{*}\$\$\alph'^{*}=\lao'\$\$\lat'$ along with
the structures of $\al$ and $\ar$ imply $\al\cap\alph'^{*}\$=\lao'\$$
and $\ar\cap\$\alph'^{*}=\$\lat'$. Thus $\al=\lao\cup\lao'\$$ and
$\ar=\lat\cup\$\lat'$. 
\end{proof}
\begin{prop}
$A$ is decomposable if and only if $\lan=\lao\lat$.
\end{prop}

\begin{proof}
If $A$ is decomposable, we know $A=A_{1}A_{2}$ as defined in \prettyref{lem:Adecomposition}.
Since $A\cap\alph^{*}=L$, $A_{1}\cap\alph^{*}=\lao$, $A_{2}\cap\alph^{*}=\lat$
and $A=A_{1}A_{2}$, $\lan=\lao\lat$ holds.

If on the other hand $\lan=\lao\lat$, obviously $A=A_{1}A_{2}$ as
in \prettyref{lem:Adecomposition}.
\end{proof}
So accordingly we get that the $\conphard$ problem $\coneq$ is reducible
to the complement of $\primfin$. Therefore the original problem $\primfin$
is $\nphard$.
\end{proof}

\section{Final remarks\label{sec:Final-remarks}}

There are still many open questions related to $\primfin$. Obviously
the exact complexity has to be determined. Our attempts to find an
$\np$-algorithm failed, so perhaps the lower bound has to be improved
further. A $\conp$ lower bound for primality with a list as input
would strongly hint to a higher lower bound for $\primfin$. Having
the input as an $\nfa$ would be yet another problem to consider.
In that case basically nothing is known, since \prettyref{thm:partitionSets}
is not applicable in its current form.

Aside from these variants for the input, the decomposition into three,
four or generally into $m$ languages is a problem to consider (for
all the input variants). A priori we do not know much about that.
For lists we still get $\conp$ for fixed $m$ by naively guessing
the partition. And for $\dfa$s we can check, for all possible decompositions
into two languages, whether those languages are decomposable again.
That approach clearly is not efficient.

Comprehensively we can say that there are still many open questions
regarding the complexity of decompositions of finite languages.

\cleardoublepage\newpage{}
\begin{acknowledgement*}
I want to thank Prof.~Dr.~Wim Martens for his guidance and support,
Dr.~Matthias Niewerth for his advice, especially for the idea to
``count'' the tiles, and Johannes Doleschal for proofreading.
\end{acknowledgement*}
\bibliographystyle{amsalpha}
\phantomsection\addcontentsline{toc}{section}{\refname}\bibliography{BAInfo}

\end{document}